\documentclass[conference]{IEEEtran}
\IEEEoverridecommandlockouts
\usepackage{cite}
\usepackage{amsmath,amssymb,amsfonts}
\usepackage{algorithmic}
\usepackage{graphicx}
\usepackage{textcomp}
\usepackage{xcolor}
\usepackage{units}
\usepackage{multirow}
\usepackage{colortbl}
\usepackage{soul}
\usepackage[left=0.625in, right=0.625in,top=0.75in,bottom=1in]{geometry}

\usepackage{amsthm}
\usepackage{thmtools, thm-restate}
\usepackage{algorithm}
\newcommand{\Cb}{\mathbb{C}}

\newcommand{\g}{\mathbf{g}}
\newcommand{\Hb}{\mathbf{H}}
\newcommand{\h}{\mathbf{h}}
\newcommand{\hh}{\widehat{h}}

\newcommand{\Rb}{\mathbb{R}}

\newcommand{\sigt}{\widetilde{\sigma}}

\newcommand{\w}{\mathbf{w}}

\newcommand{\x}{\mathbf{x}}

\newcommand{\y}{\mathbf{y}}
\newcommand{\z}{\mathbf{z}}

\newcommand{\0}{\mathbf{0}}

\newcommand{\mc}[1]{\mathcal{#1}}
\newcommand{\mb}[1]{\mathbb{#1}}
\newcommand{\mf}[1]{\mathbf{#1}}

\newtheorem{assum}{Assumption}
\newcommand{\The}{\mathbf{\Theta}}
\newcommand{\theb}{\boldsymbol{\theta}}
\newcommand{\phib}{\boldsymbol{\phi}}

\newcommand{\alg}{$\mathsf{ROAR-Fed}~$}
\newcommand{\algc}{$\mathsf{CHARLES}$}
\newcommand{\algns}{$\mathsf{ROAR-Fed}$}

\usepackage{hyperref}
\usepackage[hyphenbreaks]{breakurl}
\def\BibTeX{{\rm B\kern-.05em{\sc i\kern-.025em b}\kern-.08em
    T\kern-.1667em\lower.7ex\hbox{E}\kern-.125emX}}
\begin{document}

\title{ROAR-Fed: RIS-Assisted Over-the-Air Adaptive Resource Allocation for Federated Learning 
\thanks{This work is supported in part by NSF CNS-2112471.}
}

\author{Jiayu Mao and Aylin Yener
\\ INSPIRE@OhioState Research Center 
\\Dept. of Electrical and Computer Engineering
\\ The Ohio State University
\\ mao.518@osu.edu, yener@ece.osu.edu 
}

\newgeometry{left=0.625in, right=0.625in,top=0.75in,bottom=1in}

\maketitle

\begin{abstract}

Over-the-air federated learning (OTA-FL) integrates communication and model aggregation by exploiting the innate superposition property of wireless channels. The approach renders bandwidth efficient learning, but requires care in handling the wireless physical layer impairments.  In this paper, federated edge learning is considered for a network that is heterogeneous with respect to client (edge node) data set distributions and individual client resources, under a general non-convex learning objective. We augment the wireless OTA-FL system with a Reconfigurable Intelligent Surface (RIS) to enable a propagation environment with improved learning performance in a realistic time varying physical layer. Our approach is a cross-layer perspective that jointly optimizes communication, computation and learning resources, in this general heterogeneous setting. We adapt the local computation steps and transmission power of the clients in conjunction with the RIS phase shifts. The resulting joint communication and learning algorithm, RIS-assisted Over-the-air Adaptive Resource Allocation for Federated learning (ROAR-Fed) is shown to be convergent in this general setting. Numerical results demonstrate the effectiveness of ROAR-Fed under heterogeneous (non i.i.d.) data and imperfect CSI, indicating the advantage of RIS assisted learning in this general set up.
\end{abstract}


\vspace{0.1in}
\section{Introduction}
\label{sec:intro}
\vspace{0.1in}

In recent years, federated learning (FL) \cite{mcmahan2017} has received significant attention and found numerous applications, as a distributed machine learning framework.
FL involves an iterative training that is coordinated by a parameter server (PS) and a potentially large number of clients without sharing any of their data with the PS.
In each iteration, the clients train their local models using their individual datasets, send to PS, and the PS aggregates these local models to update the global model.
A naturally promising framework for mobile edge networks, care must be exercised when deploying FL in wireless networks, which are subject to mobile channels and limited resources. 

Over-the-air federated learning (OTA-FL)~\cite{amiri2020} provides a viable design for wireless FL by utilizing the inherent superposition property of the wireless medium. Specifically, OTA-FL proposes simultaneous analog transmissions by all participating clients for model updates over the wireless channel, which can lead to the PS directly receiving the aggregated model. Naturally, over-the-air model aggregation relies critically on the channel state information (CSI) at the transmitter.
In real systems, only estimated CSI would be available at the clients.
Imperfect CSI results in signal misalignment and can degrade the learning performance.
This paper explicitly considers a learning system under imperfect CSI.

In the context of smart and programmable radio environments, reconfigurable intelligent surfaces (RIS) \cite{wu2019intel} have emerged as a cost-effective technology to facilitate high reliability and spectral efficiency for the next generation, i.e, 6G. An RIS is typically a flat meta-surface consisting of a large number of reflecting elements, each of which is controlled to adjust phase shifts and (possibly) amplitude of the incident signal \cite{wu2019intel}.
With judicious deployment and alignment, RIS can construct the reflecting signals to desired directions, thus altering the propagation environment to a more favorable one. RIS has the potential to be integrated with edge learning to enhance the model aggregation stage, thereby boosting learning performance \cite{yang2020fed}. As such, several recent references considered RIS-augmented federated learning. In \cite{zheng2022balancing}, the mean-squared error (MSE) of the aggregated model is minimized by jointly optimizing beamformers and RIS phase shifts under both perfect CSI and imperfect CSI. \cite{liu2021csit} employs RIS to achieve OTA-FL model aggregation without CSI. \cite{wang2021fed} maximizes the number of scheduled devices under MSE constraints.
\cite{ni2021fed} adopts multiple RISs to further enhance model uploading and solves a MSE minimization problem.
In~\cite{li2022one}, RIS is applied to aid a one-bit communication FL system.
\cite{battiloro2022dynamic} jointly allocates communication and computation resources to minimize average power consumption in an RIS-assisted OTA-FL system.

More recently, unified communication and learning approaches are developed. Notably, \cite{liu2021risfl} formulates a unified communication-learning optimization problem to jointly design device selection, beamformer and RIS phase, but considers a static time-invariant channel with perfect CSI. In~\cite{zhao2022performance}, the optimality gap minimization problem of RIS-assisted FL is considered using a Lyapunov optimization framework.

Different than most existing works that focus on MSE minimization, in this paper, building on our previous work for OTA-FL without RIS assistance~\cite{mao22}, we develop a cross-layer algorithm that jointly optimally designs the communication and computation resources simultaneously to boost the learning performance in an RIS-assisted OTA-FL system. Different than existing joint communication and learning approaches, we consider the realistic time varying physical layer and imperfect CSI at the clients. We consider a general set up with a non-convex learning objective, and  heterogeneous client resources and local data distributions, aiming to demonstrate the improvement with the aid of even one RIS deployed between the clients and the PS. 
Specifically, we propose a joint communication and learning algorithm called~\alg (\underline{R}IS-assisted \underline{O}ver-the-air \underline{A}daptive \underline{R}esource Allocation for \underline{Fed}erated learning), which adapts the local update steps, transmit power and RIS phase shifts in concert in each global iteration to mitigate the impacts of both time-varying imperfect CSI and system/data heterogeneity.
We provide a convergence analysis of~\alg and observe that it achieves high test accuracy outperforming the state-of-the-art~\cite{liu2021risfl} with non-i.i.d. data and imperfect CSI.

\vspace{0.1in}
\section{System Model} \label{sec: prelim}
\vspace{0.1in}
\subsection{Federated Learning Model} 
\label{subsec: fl}
We consider a federated learning (FL) system consisting of a parameter server (PS) and $m$ clients. Client $i$ has local dataset $D_i$, sampled from distribution $\mc{X}_i$. FL minimizes the global empirical loss function by iterative collaborative training:
\begin{equation}
    \min_{\w \in \mathbb{R}^d}F(\w) \triangleq \min_{\w\in\mathbb{R}^d} \sum_{i \in [m]} \alpha_i F_i(\w, D_i), 
    \label{eq: objective}
\end{equation}
where $\w$ is the d-dimensional model, $\alpha_i = \frac{| D_i |}{\sum_{i \in [m]} | D_i |}$ is the model weight of  client $i$, $F_i(\w, D_i) \triangleq \frac{1}{| D_i |} \sum_{\xi^i_j \in D_i} F(\w, \xi^i_j)$ is the local objective, and $\xi^i_j$ is the $j$-th sample from $D_i$. $\mc{X}_i \neq \mc{X}_j$ if $ i \neq j, \forall i, j \in [m]$, i.e., we consider local datasets that are non-i.i.d., as is the case in practice.
We consider general non-convex objective functions $F_i(\w, D_i)$.
Clients in general have different volumes of training data, $\alpha_i \ne \alpha_j$ if $i\ne j$.

In FL, clients update their local models by optimizing $F_i$ and transmit them to the PS. PS aggregates the received local parameters and updates the global model accordingly. Note that in OTA-FL, aggregation and communication happen simultaneously due to the inherent superposition property of the wireless channel when all clients transmit their local model updates at the same time.
Once one communication round is completed, the PS broadcasts the current global model to the clients, and the next round starts.
When the global model converges, the training process concludes.

Specifically, in the $t$-th round, with global model $\w_t$, client $i$ computes its local gradient with its local dataset $D_i$, and performs stochastic gradient descent (SGD).
Each client $i$ trains for $\tau_t^i$ steps with an initialization of $\w^i_{t, 0} = \w_t$:
\begin{equation}
    \w^i_{t, k+1} = \w^i_{t, k} - \eta_t \nabla F_{i}(\w^i_{t, k}, \xi^i_{t, k}), \quad k = 0,\ldots,\tau_t^i-1, \label{equ:sgd}
\end{equation}
where $k$ denotes the local step and $\xi^i_{t, k}$ is the random data sample. 
The number of local steps $\tau_t^i$ varies each round and across clients, as in our previous works \cite{yang22,mao22}.

\subsection{RIS-Assisted Communication Model} 
\label{subsec: comm}
\begin{figure}[t] 
    \centering
    \includegraphics[scale=0.35]{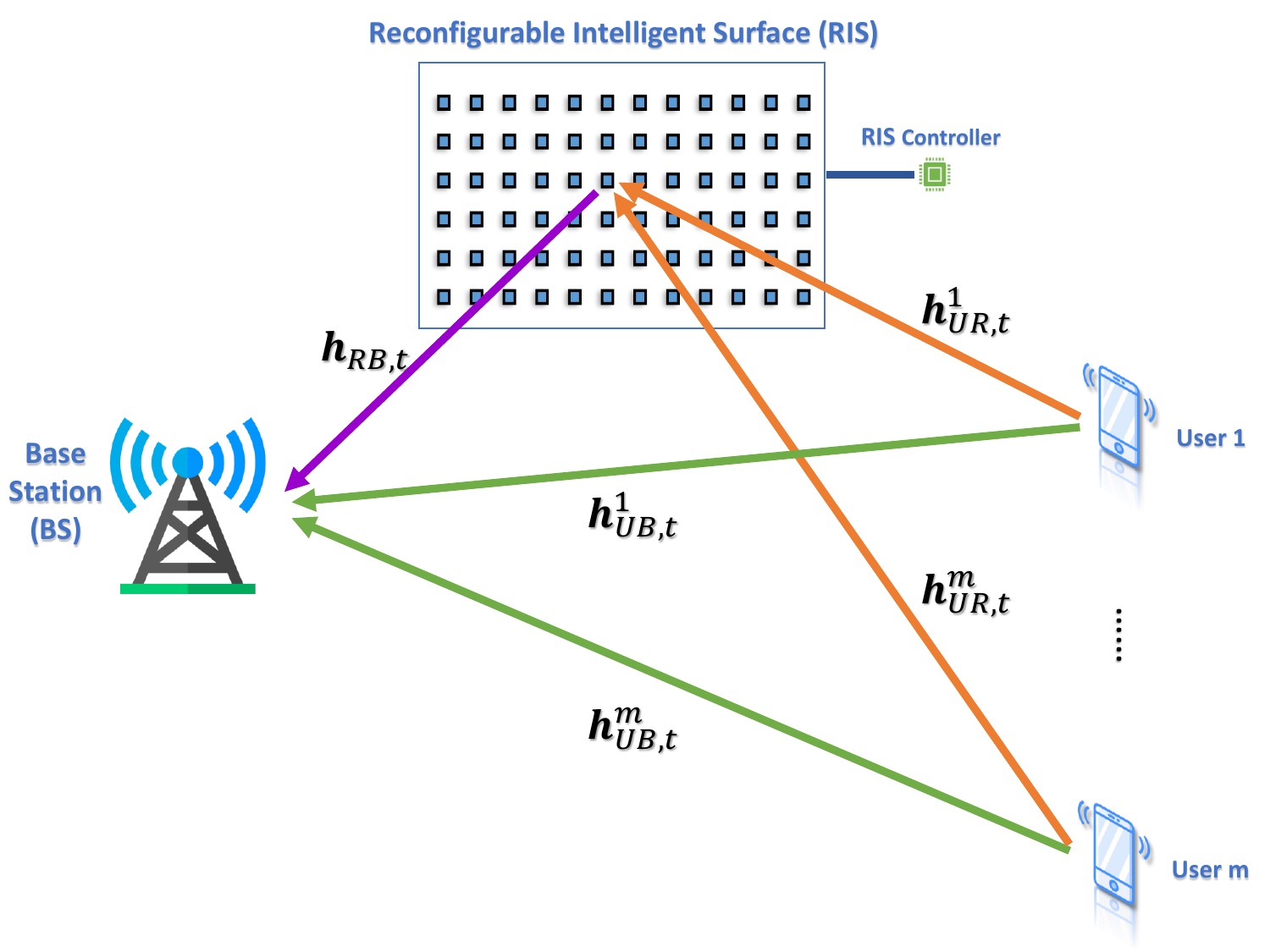}
    \caption{The RISs-assisted communication system.}
    \label{fig:sysmodel}
\end{figure}
We consider an RIS-assisted uplink communication model \footnote{Without loss of generality, we consider synchronous models.} as shown in Fig.~\ref{fig:sysmodel}, which has one RIS equipped with $N$ passive elements, $m$ single-antenna edge devices and a single-antenna base station/PS.  
The RIS is deployed between the users and the PS to aid the communication of local updates.
We assume that the direct links are weak, which renders the assistance of RIS essential.
We consider an error-free downlink, i.e., each client receives the global model perfectly, i.e., $\w^i_{t, 0} = \w_t, \forall i\in [m]$.
We assume the uplink channels follow a block fading model,
where channel coefficients remain constant for each communication round but vary independently from one round to another.
Let $\h_{UR,t}^{i} \in \Cb^N$, $\h_{RB,t} \in \Cb^N$, $h_{UB,t}^{i} \in \Cb$ denote the channels from user $i$ to RIS, from RIS to PS and from user $i$ to PS, respectively.
We represent RIS phase matrix in the $t$-th round as a diagonal matrix $\The_t = diag (\theta_{1,t}, \theta_{2,t},\cdots, \theta_{N,t})$, where $\theta_{n,t} = e^{j \phi_{n,t}}$ is the $n$-th continuous reflecting element.
Note that we update the RIS phase shifts in each global iteration.
The received signal $\y_t$ at the PS can be expressed as:
\begin{equation}
    \y_t = \sum_{i \in [m]} (h_{UB,t}^{i} + (\h_{UR,t}^{i})^H \The_t \h_{RB,t})\x^i_t + \z_t, \label{equ:receivesig}
\end{equation}
where $\x_t^i \in \mb{R}^d$ is signal from client $i$, $\z_t$ is the i.i.d. additive white Gaussian noise with zero mean and variance $\sigma_c^2$.
Define $\g_t^i$ as the cascaded user $i$-RIS-PS channel, i.e., $\g_t^i = ((\h_{UR,t}^{i})^H \Hb_{RB,t})^H \in \Cb^N$, where $\Hb_{RB,t} = diag(\h_{RB,t})$. 
Denote $\theb_t=(\theta_{1,t},...,\theta_{N,t})^T$ as the phase vector. 
Then, we can equivalently write the received signal as:
\begin{equation}
    \y_t = \sum_{i \in [m]} (h_{UB,t}^{i} +  (\g_t^i)^H \theb_t)\x^i_t + \z_t.
\end{equation}

The power constraint for user $i$ in $t$-th communication round is:
\begin{equation}
\mathbb{E}[\| \x^i_t \|^2] \leq P_t^i, \forall i \in [m], \forall t, \label{inequ:power}
\end{equation}
where $P_t^i$ is the maximum transmit power.
We assume estimated channel state information (CSI) at the clients.
We denote the estimated CSI of each path in iteration $t$ as $\hh_t$:
\begin{equation}
\hh_t = h_t + \Delta_t, \forall t,
\end{equation}
where $\Delta_t$ is the i.i.d. channel estimation error with zero mean and variance $\sigt_h^2$. All links have channel estimation error.

\vspace{0.1in}
\section{Joint Communication and Learning Design} 
\label{sec: alg}
\vspace{0.1in}

\begin{algorithm}[t!] 
    \caption{RIS-assisted over-the-air adaptive resource allocation for federated learning (ROAR-Fed)} \label{alg:apaf} 
    \begin{algorithmic}[1]
    \STATE 
    \emph{\bf Initialization: global model $\w_0$, $\theb_0$, $\beta_t^i$, $\tau_t^i, i \in [m]$.}
    \FOR{$t=0, \dots, T-1$}
    \STATE {Server first finds the client with the maximum number of local steps in round $t-1$, then applies SCA to compute the RIS phase update:} 
        \FOR{$j=0, \dots, J-1$}
        \STATE {Server updates RIS phase design $\phib_t$ by~\eqref{equ:phase}}.
        \ENDFOR
        \STATE {Server broadcasts the global model $\w_t$.}
        \FOR{each user $i \in [m]$}
        \STATE {Each client finds $\tau_t^i$ to satisfy the power constraint~\eqref{inequ:power} and trains local model by~\eqref{equ:sgd}.} 
        \STATE {Each user designs $\beta_t^i$ by~\eqref{equ:betaiimp} and transmits $\x_t^i$} by~\eqref{equ:transig}.  
        \ENDFOR
    \STATE {The server aggregates and updates global model by~\eqref{equ:globalup}.}
    \ENDFOR
    \end{algorithmic}
\end{algorithm}

In this section, we introduce a joint communication and learning design to enhance learning performance.
In each global iteration, we first update the RIS phase vector, then design the number of local steps (learning) to satisfy transmit power constraint (communication).
The overall procedure is summarized in Algorithm~\ref{alg:apaf}. 

\subsection{Dynamic Power Control}
\label{subsec:pc}
As in our previous works \cite{yang22,mao22}, we consider a dynamic power control (PC) scheme for both the server and clients.
During round $t$, user $i$ computes its signal $\x_t^i$ at the end of its local training, then transmits it to the server.
Denote $\beta_t^i$, $\beta_t$ as the adaptive PC scaling factor of user $i$ and PS, respectively.
Then the transmit signal $\x_t^i$ is designed as:
\begin{equation}
    \x_t^i = \beta_t^i (\w^i_{t, \tau_t^i} - \w^i_{t, 0}). \label{equ:transig}
\end{equation}
The PS scales the received signal~\eqref{equ:receivesig} with $\beta_t$.
Thereby, the global model is updated as:
\begin{align}
    \w_{t+1} &= \w_{t} + \frac{1}{\beta_t}\sum_{i=1}^{m} h_t^i \x_t^i  + \tilde{\z}_t, \label{equ:globalup}
\end{align}
where $\tilde{\z}_t \sim \mc{N}(\0, \frac{\sigma_c^2}{\beta_t^2} \mf{I}_d)$. 

We aim to mitigate the impact of channel fading with an appropriate PC scheme.
A well-known method is to invert the channel by CSI at the transmitter side.
For ease of notation, we use $h_t^i$ to represent the overall channel gain of client $i$ in the $t$-th round:
\begin{equation}
    h_t^i=h_{UB,t}^{i} +  (\h_{UR,t}^{i})^H \The_t \h_{RB,t}.
\end{equation}
Similarly, we use $\hh_t^i$ to represent the overall estimated CSI at the client $i$:
\begin{equation} \label{equ:csiest}
    \hh_t^i=\hh_{UB,t}^{i} +  (\hat{\h}_{UR,t}^{i})^H \The_t \hat{\h}_{RB,t}.
\end{equation}
We first consider perfect CSI.
For user $i$, we have:
\begin{equation} \label{equ:betaiperf}
    \beta_t^i = \frac{\beta_t \alpha_i}{\tau^i_t h_t^i},
\end{equation} 
\begin{equation} \label{inequ:betai}
    3\eta_t^2 \beta_t^i \tau_t^i G^2 \leq P_t^i,
\end{equation}
where $G$ is the bound of the stochastic gradient, defined in Assumption~\ref{a_bounded} in Sec.~\ref{sec: conv}. 
By~(\ref{equ:betaiperf}), the design fully offsets the impact of fading and exploits local computation resources by dynamic local steps. 
(\ref{inequ:betai}) can also facilitate finding the RIS phase update in round $t$, which will be explained in the next subsection.
In addition, this design criteria contributes to the convergence of~\alg, which is illustrated in Sec.~\ref{sec: conv}.
Finally, once the phase design is completed, each client selects $\tau_t^i$ by plugging~(\ref{equ:betaiperf}) and~(\ref{equ:transig}) into the transmit power constraint, as shown in Algorithm~\ref{alg:apaf}.

We now discuss power control for imperfect CSI.
We use the estimated channel information in~(\ref{equ:csiest}) for $\beta_t^i$:
\begin{equation} \label{equ:betaiimp}
    \beta_t^i = \frac{\beta_t \alpha_i}{\tau^i_t \hh_t^i}.
\end{equation}
Criterion~(\ref{inequ:betai}) still holds.
Note that imperfect CSI will cause signal misalignment in aggregation at the server in each iteration, thus its accumulated impact will degrade the learning performance.
We will analyze it in Sec.~\ref{sec: conv}.

\subsection{Phase Design}
\label{subsec:phase}
As mentioned in Sec.~\ref{subsec:pc}, we update the RIS phase shifts according to~(\ref{inequ:betai}).
However, there is only a single RIS to assist the communication, while the constraint~(\ref{inequ:betai}) is for each user.
Thus, in each round, we select the client with the maximum number of local steps from the previous round and update the RIS phase according to its design rule.
First, we plug~(\ref{equ:betaiperf}) into~(\ref{inequ:betai}) and obtain the inequality below:
\begin{equation} \label{inequ:phase}
    (\g_t^i)^H \theb_t \geq \frac{3 \eta_t^2 \beta_t \alpha_i G^2}{P_t^i} - h_{UB,t}^i, 
\end{equation}
where $\g_t^i$ and $\theb_t$ are cascaded RIS-assisted channel and phase vector, respectively, defined in Sec.~\ref{subsec: comm}.
To get the desired $\theb_t$, we formulate the phase design problem as:
\begin{equation} \label{prob:phase}
    \begin{aligned} 
        & \mathop{min}\limits_{\theb_t} \quad \|(\g_t^i)^H \theb_t - \frac{3 \eta_t^2 \beta_t \alpha_i G^2}{P_t^i} + h_{UB,t}^i \|_2^2  \\
        &\begin{array}{ll}
        s.t. & |\theta_{t,n}|=1 , \quad n=1,...,N.
        \end{array}
    \end{aligned}
\end{equation}

The problem in~(\ref{prob:phase}) is non-convex due to the constraint on RIS elements.
To tackle this challenge, we approximately solve it by finding a stationary solution via 
successive convex approximation (SCA)~\cite{scutari2013,mao22papa}.
The principle of SCA is to iteratively solve a sequence of simpler convex approximation problems.
Particularly, the surrogate functions are required to be strongly convex and differentiable~\cite{scutari2013}.

We define the objective function and expand it as follows:
\begin{equation} \label{equ:phasefunc}
\begin{array}{ll}
     f(\theb_t) & = || s_t^i - (\g_t^i)^H \theb_t||_2^2   \\
     & = (s_t^i)^* s_t^i - 2 Re \{ \theb_t^H \textbf{v}\} + \theb_t^H \textbf{U} \theb_t ,
\end{array}
\end{equation}
where $s_t^i = \frac{3 \eta^2 \beta_t \alpha_i G^2}{P_t^i} - h_{UB,t}^i$, $\textbf{v} = s_t^i \g_t^i$, $\textbf{U} = \g_t^i (\g_t^i)^H$.
Then we replace each phase element as $\theta_{n,t} = e^{j \phi_{n,t}}, \phi_{n,t} \in \Rb$.
Note that $s_t^i$ is a constant, hence it is equivalent to minimize
\begin{equation}
    f_1(\phib_t) = (e^{j\phib_t})^H \textbf{U} e^{j\phib_t} - 2 Re\{(e^{j\phib_t})^H \textbf{v}\},
\end{equation}
where $\phib_t = (\phi_{1,t},...,\phi_{N,t})^T$.

We then employ the SCA method. 
We set the surrogate function of $f_1(\phib_t)$ via the second order Taylor expansion at point $\phib_i^j$ in iteration $j$:
\begin{equation}
\begin{array}{ll}
      g(\phib_t,\phib^j_t) & = f_1(\phib^j_t) + \nabla f_1(\phib^j_t)^T (\phib_t - \phib^j_t)\\
     &  + \frac{\lambda}{2} ||\phib_t - \phib^j_t||_2^2,
\end{array}
\end{equation}
where $\nabla f_1(\phib^j_t)$ is the gradient, and $\lambda$ is selected to satisfy the requirement of the surrogate function, i.e, $g(\phib_t,\phib^j_t) \geq f_1(\phib_t) $. 
As such, the update rule of $\phib_t$ is:
\begin{equation}
    \phib_t^{j+1} = \phib^j_t - \frac{\nabla f_1(\phib^j_t)}{\lambda}. \label{equ:phase}
\end{equation}
When the SCA process completes, we obtain the phase update $\theb_t= e^{j\phib_t}$. 

We summarize our joint communication and learning design in Algorithm~\alg.
The design maintains the advantages of previous work~\cite{mao22}, where dynamic power control allows users to simultaneously meet communication constraints and engage in the training process.
Additionally, we integrate the RIS phase design with the learning procedure in each iteration to enhance the overall learning performance under both perfect and imperfect CSI cases.

\vspace{0.1in}
\section{Convergence Analysis}
\label{sec: conv}
\vspace{0.1in}

We first make the following assumptions on loss function:
\begin{assum}(Gradient is $L$-Lipschitz Continuous) \label{a_smooth}
	There exists a constant $L > 0$, such that $ \| \nabla F_i(\w_1) - \nabla F_i(\w_2) \| \leq L \| \w_1 - \w_2 \|$, $\forall \w_1, \w_2 \in \mathbb{R}^d$, and $i \in [m]$.
\end{assum}

\begin{assum}(Unbiased Local Stochastic Gradients and Bounded Variance) \label{a_unbias}
	The local stochastic gradient is unbiased and has a bounded variance, i.e.,
	$\mathbb{E} [\nabla F_i(\w, \xi_i)] = \nabla F_i(\w)$, $\forall i \in [m]$, and $\mathbb{E} [\| \nabla F_i(\w, \xi_i) -  \nabla F_i(\w) \|^2] \leq \sigma^2$, where $\xi_i$ is a random sample in $D_i$ and the expectation is with respect to the local data distribution $\mc{X}_i$.
\end{assum}

\begin{assum}(Bounded Stochastic Gradient) \label{a_bounded}
	There exists a constant $G \geq 0$, such that $\mathbb{E} [\| \nabla F_i(\w, \xi_i) \|^2] \leq G^2$, $\forall i \in [m]$. That is, the norm of each local stochastic gradient is bounded.
\end{assum}

With Assumptions~\ref{a_smooth}-~\ref{a_bounded}, we provide the convergence analysis of \alg as follows:

\begin{restatable}[Convergence Rate of \algns] {theorem} {convergence} \label{thm:convergence}
    Denote $\{ \w_t \}$ as a global model parameter.
    With Assumptions~\ref{a_smooth}-~\ref{a_bounded}, a constant learning rate $\eta_t = \eta  \leq \frac{1}{L}$, and $P_t^i=P_i, \forall t \in [T]$, we have:
    \begin{multline}
        \min_{t \in [T]} \mb{E} \| \nabla F(\w_t) \|^2 \leq \underbrace{\frac{2 \left(F(\w_0) - F(\w_{*}) \right)}{T \eta}}_{\mathrm{optimization \, error}} + \underbrace{\frac{L \sigma_c^2}{ \eta \beta^2}}_{\substack{\mathrm{channel \,
        noise} \\ \mathrm{error}}}   \nonumber \\
        + \underbrace{ \frac{2 m L^2}{9 \eta^2 G^2}  \sum_{i=1}^m  \frac{(\alpha_i)^2 P_i^2}{(\beta_i^2)} }_{\mathrm{local \, update \, error}} + \underbrace{  L \eta \sigma^2 \frac{1}{T} \sum_{t=0}^{T-1} \sum_{i=1}^{m} \alpha_i^2 \mb{E}_t  \bigg\| \frac{h_t^i}{\hh_t^i} \bigg\|^2}_{\mathrm{statistical \, error}} \nonumber \\
        + \underbrace{2 m G^2  \frac{1}{T} \sum_{t=0}^{T-1} \sum_{i=1}^m (\alpha_i)^2 \mb{E}_t  \bigg\| 1 - \frac{h_t^i}{\hh_t^i} \bigg\|^2}_{\mathrm{channel \, estimation \, error}} \nonumber ,
    \end{multline}
    where $\frac{1}{\beta_i^2} = \frac{1}{T} \sum_{t=0}^{T-1} \frac{1}{(\beta_t^i)^2}$ and $\frac{1}{\bar{\beta}^2} = \frac{1}{T} \sum_{t=0}^{T-1} \frac{1}{\beta_t^2}$.
\end{restatable}

\begin{proof}[Proof Highlights] 
Early steps of the proof are similar to~\cite{mao22}.
The update rule of the global model is:
\begin{equation}
    \w_{t+1} - \w_{t} = \sum_{i=1}^{m} \frac{\beta_t^i}{\beta_t} h_t^i \left(\w^i_{t, \tau^i_t} - \w^i_{t, 0}\right) + \tilde{\z}_t.
\end{equation}
According to Assumption~\ref{a_smooth}, one-step loss function descent is:
\begin{multline}
    \mb{E}_t [F(\w_{t+1})] - F(\w_t) \leq \left< \nabla F(\w_t), \mb{E}_t \left[\w_{t+1} - \w_t \right] \right> \\
    + \frac{L}{2} \mb{E}_t \left[\| \w_{t+1} - \w_t \|^2 \right].
\end{multline}
Then, by expanding each term, decoupling the channel noise and using Cauchy-Schwartz inequality, we get:
\begin{multline}
    \mb{E}_t [F(\w_{t+1})] - F(\w_t) \leq - \frac{1}{2} \eta_t \| \nabla F(\w_t) \|^2  + \frac{L \sigma_c^2}{2\beta_t^2} + \\
    \frac{L \eta_t^2}{2} \sum_{i=1}^{m} \mb{E}_t \bigg\| \frac{\alpha_i}{\tau^i_t}\frac{h_t^i}{\hh_t^i}  \sum_{k=0}^{\tau^i_t-1} \left(\nabla F_i(\w^i_{t, k}, \xi^i_{t, k}) - \nabla F_i(\w^i_{t, k})\right) \bigg\|^2  \\
    + \frac{1}{2} \eta_t \mb{E}_t \bigg\| \sum_{i=1}^{m} \frac{\alpha_i}{\tau^i_t} \sum_{k=0}^{\tau^i_t-1} \left(\nabla F_i(\w_t) - \frac{h_t^i}{\hh_t^i} \nabla F_i(\w^i_{t, k})\right) \bigg\|^2
    \label{inequ:funcdes}
\end{multline}
Note that the learning process and channel estimation are independent, the third term on the right hand side in equation~\eqref{inequ:funcdes} can be bounded as $\frac{ L \eta_t^2}{2} \sum_{i=1}^m (\alpha_i)^2 \mb{E}_t  \bigg\| \frac{h_t^i}{\hh_t^i} \bigg\|^2 \sigma^2$.
We apply constraint~\eqref{inequ:betai} and Jensen's inequality to the last term of~\eqref{inequ:funcdes} and obtain its upper bound as 
\begin{equation}
    \frac{m L^2}{9 \eta_t G^2} \sum_{i=1}^m \left( \frac{\alpha_i P_t^i}{\beta_t^i}\right)^2 + \eta_t m \sum_{i=1}^m (\alpha_i)^2 \mb{E}_t \bigg\|1- \frac{h_t^i}{\hh_t^i}\bigg\|^2 G^2.
\end{equation}
Finally, by rearranging and defining $\eta, P_i, \frac{1}{\beta_i^2}, \frac{1}{\beta^2}$, we get the convergence upper bound.
\end{proof}

From Theorem~\ref{thm:convergence}, we see that there are five error types on the convergence upper bound: the FL optimization error, the channel noise error, the local update error due to dynamic power control coupled with data heterogeneity, the statistical error from local stochastic gradients and channel estimation error caused by imperfect CSI.
Note that in the perfect CSI scenario, i.e., $h_t^i = \hh_t^i$, the channel estimation error diminishes and the convergence analysis matches the result of~\cite{yang22}.
In the case of imperfect CSI, our design incorporates the RIS-assisted communication, providing an adjustable environment for OTA-FL to overcome deep fading scenarios.
With the joint phase and dynamic local step design, \alg can still achieve an excellent learning performance even when direct links are weak, unlike~\algc~\cite{mao22} which relies on direct links.

We further bound the statistical error and channel estimation error terms by analyzing $\frac{h_t^i}{\hat{h_t^i}}$ in Theorem~\ref{thm:convergence}. 
The channel estimation error is a small perturbation in practice.
Similar to \cite{zhu2020one}, we apply the Taylor expansion: $\frac{h_t^i}{\hh_t^i} = \frac{1}{1 + \frac{\Delta_t^i}{h_t^i}}  = 1 - \frac{\Delta_t^i}{h_t^i} + \mc{O}( (\frac{\Delta_t^i}{h_t^i})^2)$. 
We ignore the higher order terms and use the facts that CSI estimation errors are IID, and each RIS phase element has a unit norm to get the result below:

\begin{restatable}{corollary} {convergence_rate} \label{cor:convergence}
Let $|\Delta_t| \ll |h_t|, \forall t\in[T]$, $h_{UB,m} = \mathop{min}\limits_{t \in [T], i \in[m]}\{|h_{UB,t}^i|\}$, $h_{UR,a} = \mathop{max}\limits_{t \in [T], i \in[m], j \in [N]}\{|h_{UB,t,j}^i|\}$, $h_{RB,a} = \mathop{max}\limits_{t \in [T], j \in [N]}\{|h_{RB,t,j}|\}$, the convergence rate of {\alg} is bounded. The statistical error and channel estimation error are bounded by:
\begin{multline}
        L \eta \sigma^2 \frac{1}{T} \sum_{t=0}^{T-1} \sum_{i=1}^{m} \alpha_i^2 \mb{E}_t  \bigg\| \frac{h_t^i}{\hh_t^i} \bigg\|^2  \leq 
        L \eta \sigma^2 \sum_{i=1}^{m} \alpha_i^2 \left( 1 + C\right), \nonumber \\
        2 m G^2  \frac{1}{T} \sum_{t=0}^{T-1} \sum_{i=1}^m (\alpha_i)^2 \mb{E}_t  \bigg\| 1 - \frac{h_t^i}{\hh_t^i} \bigg\|^2 \nonumber \leq
        2 m G^2 \sum_{i=1}^m (\alpha_i)^2 C,
\end{multline}
where $C=\frac{\sigt_h^2 (1+N^2(h_{UR,a}^2+h_{RB,a}^2+\sigt_h^2))}{(h_{UB,m})^2}$. 
\end{restatable}

Different from \cite{mao22}, the impact of accumulated channel estimation error term now depends on the number of RIS elements $N$.
This is intuitively pleasing, as more RIS elements result in more estimated channel paths. 
As a result, the bounds in Corollary~\ref{cor:convergence} increase with $N$. 
However, these bounds are not dominant in the overall convergence upper bound because in $C$, $N$ is coupled with estimation variance, the overall influence will be much smaller than the local update error term.
We also illustrate this effect in the numerical results later.

\vspace{0.1in}
\section{Numerical Results} \label{sec: exp}
\vspace{0.1in}
We simulate a RIS-assisted multiuser learning system. 
There are $m=10$ clients and the RIS is equipped with $N=16$ elements.
The system setup follows from~\cite{liu2021risfl}, where a 3D coordinate system is considered for the locations. Specifically, the PS is located at $(-50,0,10)$ meters, the RIS is placed at $(0,0,10)$ meters. The users are uniformly distributed in the x-y plane within the range of $[-20,0]$ meters in the x dimension and the range of $[-30,30]$ meters in the y dimension. 
We consider a channel model where the small-scale fading is i.i.d. Gaussian. We adopt the path loss model from~\cite{tang2020ris}.
The path loss of user-PS direct link is $ G_{PS}G_{U}\left(\frac{3*10^8 m/s}{4 \pi f_c d_{UP}} \right)^{PL}$, where $G_{PS}=5$dBi, $G_U=0$dBi are antenna gain at the PS and user, respectively; $f_c=915$MHz is the carrier frequency; $d_{UP}$ is the distance between user and PS; $PL$ is the path loss exponent.
We set $PL=4$ to simulate weak direct links.
The path loss of RIS assisted link is $G_{PS}G_{U} G_{RIS} \frac{N^2 d_x d_y ((3*10^8 m/s)/f_c)^2 }{64 \pi^3 d_{RP}^2 d_{UR}^2}$, where $G_{RIS}=5$dBi is the RIS antenna gain; $d_x=d_y=(3*10^7 m/s)/f_c$ are the horizontal and vertical size of a RIS element; $d_{RP}, d_{UR}$ are distances between RIS and PS, user and RIS, respectively. 
We simulate the channel estimation error as a Gaussian random variable with variance $\sigt_h^2= 0.1 \sigma_c^2$.
The maximum SNR is set to \unit[$20$]{dB}.

We consider an image classification task by logistic regression on the MNIST dataset~\cite{lecun1998gradient}. 
We focus on imperfect CSI and the extremely heterogeneous case of non-i.i.d. data distribution. That is, each client contains a local dataset with only {\it one} class.
We compare~\alg with two baselines:
\vspace{0.05in}
\begin{list}{\labelitemi}{\leftmargin=1em \itemindent=-0.5em \itemsep=.2em}
\item Baseline 1: No RIS, only direct links exist (\algc).
\item Baseline 2: Algorithm from~\cite{liu2021risfl}.
\end{list}
\vspace{0.05in}
\begin{figure}[t] 
    \centering
    \includegraphics[scale=0.325]{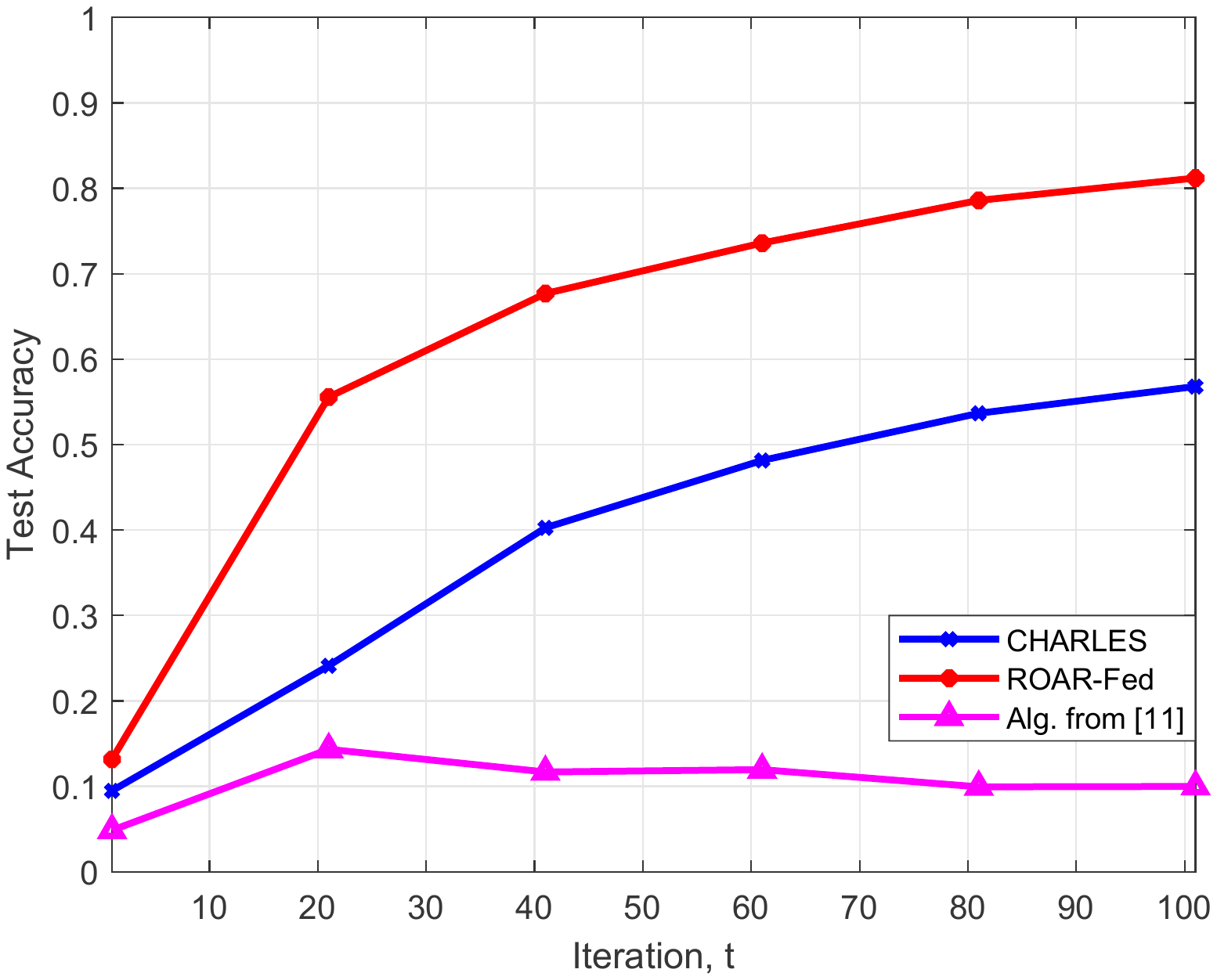}
    \caption{Test accuracy.}
    \label{fig:result1}
\end{figure}
\begin{figure}[t] 
    \centering
    \includegraphics[scale=0.325]{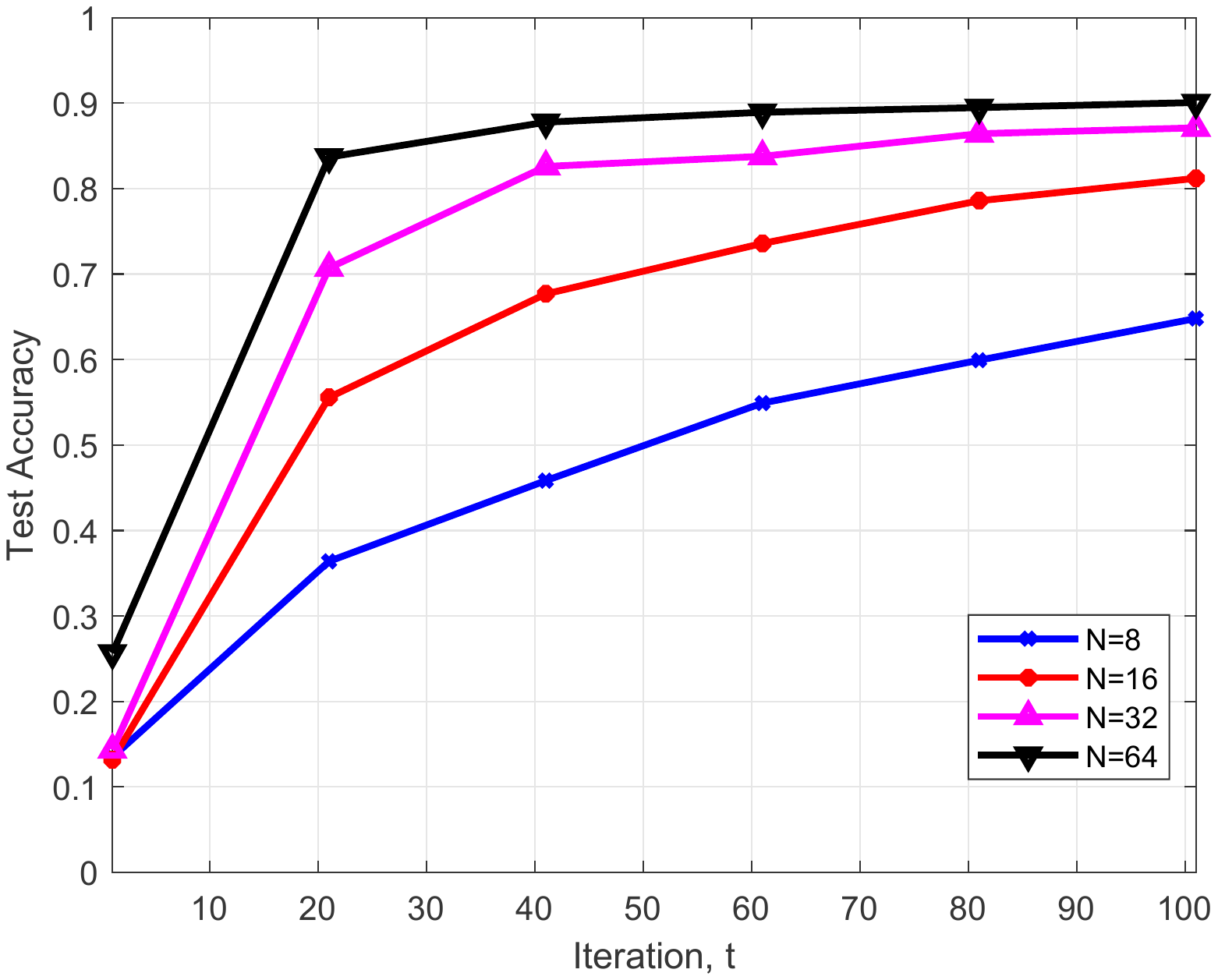}
    \caption{Test accuracy versus the number of RIS elements.}
    \label{fig:result2}
\end{figure}

In Fig.~\ref{fig:result1}, the test accuracy versus the global training round is shown.
\alg outperforms~\algc, which illustrates the merit of deploying RIS.
Under the imperfect CSI, our algorithm achieves excellent test accuracy and outperforms the state-of-the-art algorithm from~\cite{liu2021risfl} which fails to converge. 
This verifies the effectiveness and robustness of the proposed joint adaptive communication and learning design.
We may interpret that~\alg also achieves some level of fairness, in a way that for each client, the coupled adaptive RIS phase and local steps design ensures a better wireless environment and training performance for all clients equally. 
Note that in our non-i.i.d. data setting, more local updates result in a more biased global model towards that client.

In Fig.~\ref{fig:result2}, we evaluate the impact of the number of RIS elements $N$. 
We observe that with the increase of $N$, both the test accuracy and the convergence speed increase.
This matches our theoretical convergence analysis in Sec.~\ref{sec: conv}.
It is important to reiterate that when $N$ increases, the channel estimation error grows as discussed in Corollary~\ref{cor:convergence}. However, this term is not dominant in the convergence upper bound.
The local update error term is dominant, which will decrease with increasing $N$, because the better radio environment facilitated by the larger RIS will lead to fewer local steps and a larger power control factor overall. 

\vspace{0.1in}
\section{Conclusion} 
\label{sec: conclusion}
\vspace{0.1in}
In this paper, we have considered an RIS-assisted wireless edge network that employs over-the-air federated learning. We have proposed a new adaptive OTA-FL approach.
The proposed algorithm brings a cross-layer perspective in that, it jointly optimizes communication and computation resources, in particular the number of local computation steps, coupled with RIS phase design, in order to optimize a general non-convex learning objective with non-i.i.d. client distributions, and estimated (imperfect) CSI at the clients. 
We have proved the convergence of~\alg and discussed the impact of RIS with respect to the number of reflecting elements. We have illustrated the effectiveness and robustness of~\alg under heterogeneous data distributions and imperfect CSI.
Future work in this direction includes the impact of noisy downlink, programming the environment in both uplink and downlink utilizing one or more RISs.

%\newpage
%\vspace{0.5in}
\bibliographystyle{IEEEtran}{}
\bibliography{BIB/Optimization,BIB/Relay, BIB/RIS, BIB/ImperfectCSI, BIB/Yener, BIB/RISFL, BIB/Experiments, BIB/Introduction}

% \onecolumn
% \input{proof.tex}

\end{document}